\documentclass{svjour3}                     
\smartqed  
\usepackage{graphicx}
\usepackage{dcolumn}
\usepackage{bm}
\usepackage{booktabs}
\usepackage{xcolor}
\usepackage{amssymb}
\usepackage{amsmath}
\usepackage[numbers]{natbib}
 \journalname{Quantum Information Procesing}
\begin{document}

\title{Direct-dynamical Entanglement-Discord relations}


\author{Virginia Feldman       \and
        Jonas Maziero  \and
        A. Auyuanet
}


\institute{Virginia Feldman \at
             Instituto de F\'isica, Facultad de Ingenier\'ia, Universidad de la Rep\'ublica, J. Herrera y Reissig 565, 11300, Montevideo, Uruguay 
                   \and
           Jonas Maziero \at
              Instituto de F\'isica, Facultad de Ingenier\'ia, Universidad de la Rep\'ublica, J. Herrera y Reissig 565, 11300, Montevideo, Uruguay \\
                \emph{Present address:Departamento de F\'isica, Centro de Ci\^encias Naturais e Exatas, Universidade Federal de Santa Maria, Avenida Roraima 1000, 97105-900, Santa Maria, RS, Brazil}
                \and
               A. Auyuanet \at
               Instituto de F\'isica, Facultad de Ingenier\'ia, Universidad de la Rep\'ublica, J. Herrera y Reissig 565, 11300, Montevideo, Uruguay \\
             \email{auyuanet@fing.edu.uy}
}

\date{Quantum Inf. Proc.(2017)16:128 / doi:10.1007/s11128-017-1580-4}

\maketitle

\begin{abstract}
In this article, by considering Bell-diagonal
two-qubit initial states under local dynamics generated by
the Phase Damping, Bit Flip, Phase Flip, Bit-Phase Flip, and Depolarizing
channels, we report some elegant direct-dynamical relations between
geometric measures of Entanglement and Discord. The complex scenario
appearing already in this simplified case study indicates that a similarly
simple relation shall hardly be found in more general situations.
\keywords{Quantum Correlations \and Discord \and Entanglement}
\end{abstract}

\section{Introduction}
\label{intro}
                                                                                                                                                                                                                                                                                     
Since 1935 \cite{EPR,Borh,Schrodinger}, the quantumness of correlations
has triggered philosophical and scientific debates regarding
some unusual qualities of the physical world. From these discussions
grew quantum information science (QIS) \cite{nielsen2000quantum,Preskill,Wilde},
which is excitingly reverberating on several fields of scientific
inquire \cite{Zurek_Dar,OBrien,Nori,Ekert,LloydM,Jarzynski,McFadden,PetruccioneAI,PetruccioneNN,Svore,Andraca,Marvian,Preskill_15}.
There are some key concepts which are very significant for this quantum
revolution, among them are quantum contextuality \cite{Spekkens,Horodecki_Ct},
non-locality \cite{Dam,Wehner_RMP}, Steering \cite{Wiseman_St,Reid},
Entanglement \cite{Bruss_E02,Plenio_E07,Horodecki_RMP,Davidovich_RPP,GlobalEnt,sharma2016robustness,Sun2017},
Discord \cite{Zurek_D,Celeri_IJQI,Vedral_D_RMP,Streltsov}, and Coherence
\cite{Plenio_Ch,Adesso_ChE}. Understanding them and their interplay
is a fundamental current problem in QIS. \\
Although some of these quantumness measures have been indirectly related through purifications,
as they are qualitatively and operationally quite distinct, at first
sight there is no reason one should expect to exist a direct relation
between them. \\
In recent years special attention has been given to the geometric definition 
of quantifiers of correlations \cite{ShimonyEntanglement,GeometricEntanglement,BellomoGeometric,BellomoGeometric2,Dakic,GaussianDiscord,SquareNormDistanceCorr}.This geometric approach requires to choose a measure of distance, which has also been the subject of wide discussions.  
One specific property that is required for a measure of distance is its
contractivity. This feature is particularly important because it ensures that
the distance between two quantum states initially in an uncorrelated state
with its environment doesn't increase with time, which assures
that the distinguishability doesn't grow with time. The contractivity of the 
Trace norm for trace-preserving quantum evolution was proved by Ruskai \cite{Ruskaicontractivity},
but for the case of the Hilbert-Schmidt norm Osawa \cite{Osawa_2ndp} found one example in which 
this norm is not contractive. More recently, Wang and Schirmer \cite{Schirmer} demonstrated
necessary and sufficient conditions for contractivity of the Hilbert-Schmidt norm.\\
Indirect relations between Entanglement and Discord of different bipartitions
of a three-partite system were obtained in the literature \cite{Bruss_ED_M,Adesso_ED_M,Fanchini_Eirr,Fanchini_ED_cl,Yan_ED}.
However there is no straightforward reason to assume the existence of a direct relation
between Entanglement and Discord for a single bipartite system, although numerical
and analytical inequalities can be found \cite{James_ED,Adesso_EDb1,Adesso_EDb2,Vianna_ED,CampBell_ED,Angelo_ip}.\\
In this article we will consider the possible existence of a direct relation between Entanglement and Discord, two of the most prominent quantum correlations. 
In particular we calculate for initial Bell-diagonal states under the action of
several local trace-preserving quantum-channels (Bit Flip, Bit Phase Flip, Phase Damping and Depolarizing)
geometric measures for Entanglement and Discord and we find direct relations between them. 
In Sec. \ref{sec:geometricmeasurement} we present the geometric measures of correlations and the corresponding measures of distance. In Sec.\ref{sec:EntDiscBDS} we obtain the Entanglement-Discord relations for Bell-diagonal states under the Phase Damping Channel, using the Hilbert-Schmidt norm and the Trace norm. In particular for the Trace norm we find that our geometric measure of Entanglement is equal to the Concurrence  of the entangled state. In Sec. \ref{sec:EntDiscDEPO}
we obtain the Entanglement-Discord relations for Bell-diagonal states under the Depolarizing channel. Finally, in Sec.\ref{sec:conclusion} we summarize and discuss our results.

\section{Geometric measures of correlations}
\label{sec:geometricmeasurement}
Concerning the quantification of Entanglement and Discord, we choose a geometric
approach in order to set both quantifiers in equal foot,  
strategy that allows us to find a direct relation between them.
We can measure how entangled or Discordant
a general bipartite density operator $\rho$ is by its minimum distance
to the corresponding set not possessing that property, i.e.,:
\begin{eqnarray}
E(\rho) & = & \min_{\rho^{sep}}d(\rho,\rho^{sep}),\\
D(\rho) & = & \min_{\rho^{cc}}d(\rho,\rho^{cc}),
\end{eqnarray}
where $\rho^{sep}$ belongs to the set of separables states (zero Entanglement)
and $\rho^{cc}$ to the set of classical states (zero Discord).\\
In order to quantify Entanglement and Discord from a geometric approach,
we choose two particular cases of the
Schatten $p-$norm distance, which is defined using the $p$-norm:
$||A||_{p}^{p}:=\mathrm{Tr}\left(\sqrt{A^{\dagger}A}\right)^{p}$.
When $p=1$ we have the Trace norm that induces the distance $d_{1},$
\begin{equation}
 d_{1}(\rho,\sigma)= ||\rho-\sigma||_{1}
\end{equation}
and when $p=2$ we have the Hilbert-Schmidt
norm that induces the distance $d_{2},$
\begin{equation}
 d_{2}(\rho,\sigma)= ||\rho-\sigma||^2_{2}.
\end{equation}
As we discussed, the Trace norm is contractive; and using the results mentioned above \cite{Schirmer} 
we found that, for the quantum channels we work with: Phase Damping, Depolarizing, Bit Flip and Bit Phase Flip, 
the Hilbert-Schmidt norm is also contractive.

\section{Entanglement-Discord relations for Bell-diagonal
states evolving under local Phase Damping channel}

\label{sec:EntDiscBDS}

Bell-diagonal states have been the workhorse for analytical
computations of most quantum correlation functions. This is helpful,
for instance, when investigating the dynamics of these correlations
under decoherence \cite{Davidovich_PRAe,Sarandy_2SCe}, for their
experimental estimation \cite{Maziero_BJP_R,Pryde}, and to more
easily visualize and comprehend the geometrical structure of the quantum
state space \cite{Horodecki_BDS1,Horodecki_BDS2,Caves_BDS,Fan_BDS}.
The Bell-diagonal class of states is equivalent, under local unitaries,
to two-qubit states with maximally mixed marginals that can be written
as follows:
\begin{equation}
\rho=2^{-2}\left(\mathbb{I}_{4}+\vec{r}\cdot\vec{\Sigma}\right),\label{eq:BDS}
\end{equation}
with $\mathbb{I}_{n}$ denoting the $n\mathrm{x}n$ identity operator
and we define $\vec{\Sigma}=\sigma_{1}\otimes\sigma_{1}\hat{i}+\sigma_{2}\otimes\sigma_{2}\hat{j}+\sigma_{3}\otimes\sigma_{3}\hat{k}$
and $\vec{r}=r_{1}\hat{i}+r_{2}\hat{j}+r_{3}\hat{k}$. The
last vector is dubbed here as the correlation vector, because $r_{j}=\mathrm{Tr}(\rho\sigma_{j}\otimes\sigma_{j})$.
The set \{$\hat{i},\hat{j},\hat{k}$\} is the usual orthonormal basis
for $\mathbb{R}^{3}$ and $\sigma_{j}$ are the Pauli matrices. \\
In the following, we study a two-qubit system $ab$ prepared in the
state $\rho$ with each party interacting with its own local-independent
environment. Thus, if the overall initial state is $\rho\otimes|E_{0}^{a}\rangle\langle E_{0}^{a}|\otimes|E_{0}^{b}\rangle\langle E_{0}^{b}|$,
the closed system evolved state is $U_{a}\otimes U_{b}(\rho\otimes|E_{0}^{a}\rangle\langle E_{0}^{a}|\otimes|E_{0}^{b}\rangle\langle E_{0}^{b}|)U_{a}^{\dagger}\otimes U_{b}^{\dagger}$,
with $U_{s}$ being an unitary transformation in $\mathcal{H}_{s}\otimes\mathcal{H}_{E_{s}}$,
for $s=a\mbox{, }b$. The partial trace leads then to the evolved
state for the bipartite system $ab$ \cite{nielsen2000quantum,Preskill,Wilde}:
\begin{equation}
\rho_{t}={\textstyle \sum_{m,n}}K_{m}^{a}\otimes K_{n}^{b}\rho K_{m}^{a\dagger}\otimes K_{n}^{b\dagger},\label{eq:evolved}
\end{equation}
with the matrix elements of the Kraus' operators being $\langle s_{j}|K_{m}^{s}|s_{k}\rangle=(\langle s_{j}|\otimes\langle E_{m}^{s}|)U_{s}(|s_{k}\rangle\otimes|E_{0}^{s}\rangle)$,
where $|s_{j}\rangle$ and $|E_{m}^{s}\rangle$ are orthonormal basis
for $\mathcal{H}_{s}$ and $\mathcal{H}_{E_{s}}$, respectively. 
Using this, we obtain direct-dynamical Discord-Entanglement
relations considering two qubits prepared in a Bell-diagonal state
interacting with the Phase Damping channel.\\
The Phase Damping Channel represents the main source of decoherence
in quantum systems. This kind of environment obtains information about
the system's state without exchange of energy. From the phenomenological
map describing this process we can obtain the Kraus operators $K_{0}^{pd}=\sqrt{1-p_{s}}\mathbb{I}_{2}$,
$K_{1}^{pd}=\sqrt{p_{s}}|s_{0}\rangle\langle s_{0}|$, and $K_{2}^{pd}=\sqrt{p_{s}}|s_{1}\rangle\langle s_{1}|$.
Note that $p_{s}$ is the probability for a change in the environment
state, which is conditioned on the system state. From now on we will
assume symmetric environments, i.e., $p_{a}=p_{b}\equiv p$. Substituting
these Kraus operators in Eq. (\ref{eq:evolved}), we obtain the evolved
correlation vector:
\begin{equation}
\vec{r}_{pd}(p)=r_{1}(1-p)^{2}\hat{i}+r_{2}(1-p)^{2}\hat{j}+r_{3}\hat{k}.
\end{equation}
Following this evolution, the Bell-diagonal form of the evolved state is
preserved.
In Fig. \ref{fig:bds} we show the Bell-diagonal state space \cite{Caves_BDS}, its
classical-classical and separable sub-sets, and examples of trajectories
for the dynamics considered in Sec.\ref{sec:EntDiscBDS} and in Sec. \ref{sec:EntDiscDEPO}.

\begin{figure*}
\includegraphics[width=0.40\textwidth]{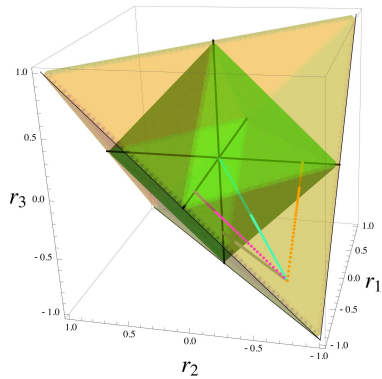}
\caption{(color online) $\mathbb{R}^{3}$ representation of the Bell-diagonal
state space. The $(r_{1},r_{2},r_{3})$ points in the tetrahedron
are the physical states. The entangled states are those in the yellow
regions. The separable convex subset is the green octahedron. The
three black axes are the classical-classical states. The dotted lines
represent the ``temporal'' trajectories for the initial state $r_{1}=r_{2}=r_{3}=-0.7$
evolved under local Phase Damping, Bit Flip, Bit Phase Flip and Depolarizing channels.}
\label{fig:bds}
\end{figure*}

\subsection{Entanglement - Discord relation using the Hilbert-Schmidt norm}

\label{sec:HSnorm}

In order to calculate the Discord we use the geometric measure introduced
by \cite{Dakic}:
\begin{equation}
D(\rho)= \min_{\chi}\Vert \rho-\chi\Vert ^{2}_{2},
\end{equation}
with the minimum $\chi$ taken over the set of the states with zero Discord. As it was noted by
\cite{Caves_BDS}, for the Bell-diagonal states the only states with zero Discord lay in the Cartesian
axes \ref{fig:bds}, thus the Discord is determined by the euclidean distance from the point $(\vert r_{1}(p)\vert ,\vert r_{2}(p)\vert ,\vert r_{3}(p)\vert )$
to its closest axis. Taking this into account we define the Discord:
\begin{equation}
D(p)=Min[D_{1}(p),D_{2}(p),D_{3}(p)], 
\end{equation}
where
\begin{eqnarray}
&& D_{i}(p)=r_{j}(p)^2+r_{k}(p)^2, i,j,k=1,2,3, \nonumber \\ 
&& i\neq j \neq k.
\end{eqnarray}
As we will discuss later, depending on the initial value of the correlation vector, the Discord
could present ``sudden changes'' during the evolution of the system, due to the fact that $D(p)$ could change between the three
different functions: $D_{1}(p)$, $D_{2}(p)$ or $D_{3}(p)$. \\
For the Entanglement calculation we use the known result \cite{Schirmer} that
for hermitian matrices which pertain to the same trace class, the Hilbert-Schmidt distance between
them equals the Euclidean distance between their corresponding correlation vectors: 
$\Vert\rho-\zeta\Vert_{2}=\Vert\vec{r}-\vec{z}\Vert_{2}.$ \\
The region of separable states is bounded by the octahedron \cite{Caves_BDS}
defined by $\vert r_{1}\vert+\vert r_{2}\vert+\vert r_{3}\vert\leq 1$. Using the symmetry of the separable states
region in the Bell-diagonal space \ref{fig:bds}, we can calculate the Entanglement for the state
with evolved correlation vector $\vec{r}(p)$ as the euclidean distance from the point
$(\vert r_{1}(p)\vert,\vert r_{2}(p)\vert,\vert r_{3}(p)\vert)$ to the plane $\vert r_{1}\vert +\vert r_{2}\vert +\vert r_{3}\vert =1$, when
$(\vert r_{1}(p)\vert +\vert r_{2}(p)\vert +\vert r_{3}(p)\vert )\geq 1$. 
In order to relate Discord Entanglement, we define:
\begin{equation}
E(\rho)= \min_{\rho^{sep}}\Vert \rho-\rho^{sep}\Vert ^{2}_{2},
\end{equation}

which for the case of Phase Damping takes the 
following expression:
\begin{equation}
 E(p)=\frac{\left(\vert r_{1}\vert (1-p)^{2}+\vert r_{2}\vert (1-p)^{2}+\vert r_{3} \vert -1\right)^2 }{3}.
\end{equation}
Depending on initial conditions we have three different scenarios:
(1) considering initial Bell-diagonal states with $\vert r_{1}\vert ,\vert r_{2}\vert < \vert r_{3}\vert $, the Discord takes 
the form $D_{3}(p)$ for every $p$. \\
(2) When $\vert r_{2}\vert ,\vert r_{3}\vert < \vert r_{1}\vert$   ($\vert r_{1}\vert \neq 0$), the Discord exhibits a sudden change for 
$p_{13}=1-\sqrt{\frac{\vert r_{3}\vert }{\vert r_{1}\vert }}$. It starts equal to $D_{1}(p)$ for $p\leq p_{13}$, and then
it becomes $D_{3}(p)$ for $p \geq p_{13}$. \\
(3) Similarly, if $\vert r_{1}\vert ,\vert r_{3}\vert < \vert r_{2}\vert$  ($\vert r_{2}\vert \neq 0$), the Discord starts as
$D_{2}(p)$ and then it becomes $D_{3}(p)$ for $p \geq p_{23}$, where $p_{23}=1-\sqrt{\frac{\vert r_{3}\vert }{\vert r_{2}\vert }}.$ \\
Taking into account that the Entanglement undergoes sudden death for $p_{SD}=1-\sqrt{\frac{1-\vert r_{3}\vert }{\vert r_{1}\vert +\vert r_{2}\vert }}$
we can express the Discord for $p \in [0,p_{SD}]$ as a function of the Entanglement in $p$ and the parameter $\vert r_{3}\vert $.\\
(1) When $D(p)=D_{1}(p)$ or $D_{2}(p)$, we have $D(p)=D_{i}(p)$ with  
\begin{equation}
D_{i}(p)=\vert r_{j}\vert ^2 \left(\frac{\sqrt{3E(p)}-\vert r_{3}\vert +1}{\vert r_{1}\vert +\vert r_{2}\vert } \right)^2 + \vert r_{3}\vert ^2 ,
\end{equation}
where $i,j=1,2$ and $i \ne j$. \\
(2) In the case that $D(p)=D_{3}(p)$, 
\begin{equation}
 D(p)= \left(\vert r_{1}\vert ^2 +\vert r_{2}\vert ^2\right) \left(\frac{\sqrt{3E(p)}-\vert r_{3}\vert +1}{\vert r_{1}\vert +\vert r_{2}\vert }\right)^2.
 \end{equation}
  In Fig.\ref{fig:graficas1} , for the case where the initial conditions verify $\vert r_{3}\vert <\vert r_{2}\vert < \vert r_{1}\vert$, we show  the curve of Discord as a function of Entanglement exhibiting one of the sudden changes described previously.

\subsection{Entanglement - Discord relation using the Trace norm}
\label{sec:TraceDistance}

In the following, we calculate Entanglement and Discord based on the trace norm:
\begin{eqnarray}
E(\rho) & = & \min_{\rho^{sep}} \Vert \rho - \rho^{sep}\Vert_{1}, \nonumber \\ 
D(\rho) & = & \min_{\rho^{cc}} \Vert \rho - \rho^{cc}\Vert_{1}.
\end{eqnarray}

\subsubsection{Entanglement}
As far as we know there is no expression for the Entanglement by means of the trace distance.
In order to calculate Entanglement we need to find the closest separable state to a Bell 
diagonal state.
There are several proposals to determine the closest separable state using different criteria
\cite{Lewenstein-Sanpera,Wellens,Ishizaka} and different measures of distance.
In this section we use a result obtained in the context of multipartite systems
\cite{Eberly_X}, which shows that the closest biseparable state to a general $X$ state
is also an $X$ state. We propose the following criterion: to search for the closest separable state to an
entangled X state over the X states with the same populations as the entangled one.
With this assumption we get a very meaningful result. 
\begin{proposition}
Let $\rho_{X}$ be the density matrix associated to a quantum $X$ state and $\sigma_{X}$ its closest separable
state (with the same populations as $\rho_{X}$) under the trace distance. 
The trace distance between $\rho_{X}$ and $\sigma_{X}$ is the Concurrence of
$\rho_{X}$.
\end{proposition}
\begin{proof} - Let's consider $\rho_{X}$, a density matrix associated to a general $X$ state: \\
$\rho_{X}=$ $\begin{pmatrix}
  a& 0& 0& e \\
  0& b& f& 0 \\
  0& f& c& 0 \\
  e& 0& 0& d
 \end{pmatrix}.$\\
To represent a physical state, $e$ and $f$ 
must fulfill the following restrictions: $\vert e\vert \leq \sqrt{ad}$ and $\vert f\vert \leq \sqrt{bc}$.
We will search for the closest separable state to $\rho_{X}$ over the set of $X$
matrices constrained to having the same diagonal as $\rho_{X}$: \\
$\sigma_{X}=$ $\begin{pmatrix}
  a& 0& 0& e' \\
  0& b& f'& 0 \\
  0& f'& c& 0 \\
  e'& 0& 0& d
 \end{pmatrix}.$\\
If $\sigma_{X}$ represents a physical and separable
state, $e'$ and $f'$ must fulfill the following condition:
\begin{equation}
 \vert e'\vert , \vert f'\vert \leq \min \{\sqrt{ad},\sqrt{bc}\}.
\end{equation}
As we are searching for the closest separable $X$ state according to the trace distance, we
have to minimize 
\begin{equation}
 \Vert\rho_{X}-\sigma_{X}\Vert_{TD}= 2 \left( \vert\,\,(\vert e\vert-\vert e'\vert)\,\,\vert + \,\,(\vert f\vert-\vert f'\vert)\,\,\vert \right),
\end{equation}
where we used the fact that the minimum is attained when the off diagonal elements of $\rho_{X}$
and $\sigma_{X}$ have the same sign.\\
The expression above reaches its minimum value when both terms are 0, but this only happens
when both matrices are equal. In the general case we can only equate to zero one term, depending
on the value of the elements of $\rho_{X}$: (1) if $\sqrt{bc}< \sqrt{ad} $ it is easy to see that the
minimum is reached when $f'= f$ and $e'=\sqrt{ad}.$ In this case we have 
\begin{equation}
 \Vert \rho_{X}-\sigma_{X}\Vert_{1}= 2 \vert\,\,\vert e\vert-\sqrt{bc}\,\, \vert .
\end{equation}
(2) If  $\sqrt{ad}< \sqrt{bc} $ the minimum is reached when $e'= e$ and $f'=\sqrt{ad}.$ In this case we
have 
\begin{equation}
 \Vert\rho_{X}-\sigma_{X}\Vert_{1}= 2 \vert \,\, \vert f\vert-\sqrt{ad}\,\vert .
\end{equation}
We can easily see that for $\sqrt{ad}=\sqrt{bc}$, the matrix $\rho_{X}$ is separable. \\ 
As demonstrated by \cite{YuEberlyConcX}, for a general $X$ state the Concurrence has the following expression:
\begin{equation}
 C(\rho_{X})= 2 \max \{0,\vert e\vert -\sqrt{bc},\vert f\vert-\sqrt{ad} \}.
\end{equation}
This bring us to our result: the trace distance between a general $X$ state and its
closest separable state is the Concurrence of the general $X$ state. \qed
\end{proof}
Taking into account this result we have for the geometric Entanglement: 
\begin{equation}
E(\rho) =  \min_{\rho^{sep}}d(\rho,\rho^{sep})=C(\rho). 
\end{equation}
This result allows us to calculate the trace norm based geometric Entanglement, even for evolutions that don't preserve the  
Bell-diagonal form, as for example $X$ states under the action of the Amplitude Damping channel.\\
At this point it is worth mentioning that for the evolution of the Bell-diagonal states under the Phase Damping channel, 
the Entanglement is determined by the sign of $r_{3}$: if $r_{3}>0$ then $C_{1}(\rho_{X})=2 (\vert e\vert -\sqrt{bc})$ and
if $r_{3}<0$, $C_{2}(\rho_{X})=2 (\vert f\vert-\sqrt{ad})$, for all the evolution until become separable. 
In a pictorical view: the Entanglement of states that belong to the upper part of the tetrahedron is described
by $C_{1}(\rho_{X})$ and the Entanglement of states that belong to the lower part of the tetrahedron is described
by $C_{2}(\rho_{X})$.
\subsubsection{Discord}
Expressions to calculate the Discord by the Schatten-1 norm for the Bell-diagonal states are known \cite{Paula_Oliveira_Sarandy},
\begin{equation}
 D=int\{|c_{1}|,|c_{2}|,|c_{3}|\},
\end{equation}
where ``int'' means the intermediate value. 
\subsubsection{Discord-Entanglement relations}
Having properly defined the geometric quantifiers for Entanglement and Discord using the Schatten-1 norm
we can relate them.
To calculate the Discord we have to establish the ordering between the absolute values of the
coefficients $|r_{i}|$ as we did when we used the Hilbert-Schmidt norm. \\ 
(1) Considering $|r_{i}|<|r_{j}|<|r_{3}|$,
where $i, j=1,2$ and $i \ne j$, it is clear that the Discord $D$ is expressed by $D=|r_{j}(p)|$ for all the evolution.\\
The Entanglement of the state quantified by the concurrence is $C(p)=\frac{1}{2} \max \{ |r_{2}\pm r_{1}|(1-p)^2-|1\pm r_{3}| \}$.
Combining the two expressions we have: 
\begin{equation}
 D(p)=\frac{2C(p)+|1 \pm r_{3} | }{|r_{2} \pm r_{1}|}|r_{j}|, \forall p \in [0,p_{SD1}],
\label{eq:Discord0}
\end{equation}
where $p_{SD1}$ is when Entanglement sudden death occurs; for this case
$p_{SD1}=1-\sqrt{\frac{|1\pm r_{3}|} {|r_{2}\pm r_{1} |}}$. \\

(2) When $|r_{1}|<|r_{3}|<|r_{2}|$, the Discord exhibits a sudden change in $p_{I}=1-\sqrt{\frac{|r_{3}|}{|r_{2}|}}$:
it starts being equal to $|r_{3}|$ for $p \leq p_{I}$ and then becomes equal to $|r_{2}(p)|$ for $p \geq p_{I} $. We note that
$p_{I}=p_{23}$, $p_{23}$ being the time for the sudden change of geometric Discord calculated using the Hilbert-Schmidt norm.\\
(3) When $|r_{3}|<|r_{2}|<|r_{1}|$, the Discord
exhibits two sudden changes, in $p_{I}$ and in $p_{II}= 1-\sqrt{\frac{|r_{3}|}{|r_{1}|}}$, verifying for
this initial conditions that $p_{I} < p_{II}$. For $p<p_{I}$ the Discord  $D$ is expressed by $D=|r_{2}(p)|$, for $p_{I}\leq p<p_{II}$
the Discord is expressed by $D=|r_{3}|$ and for $p_{II}\leq p$ 
\begin{equation}
 D(p)=|r_{1}(p)|=\frac{2C(p)+|1 \pm r_{3}|}{|r_{2} \pm r_{1}|}|r_{1}|.
\label{eq:Discord0}
\end{equation}
We note that, analogous to case (2),
$p_{II}=p_{13}$, $p_{13}$ being the time for the sudden change of the Hilbert-Schmidt norm based geometric Discord.\\

The same results are found for the Bit Flip and the Bit Phase Flip channels as
we can see from their correlation vectors:
\begin{equation}
\vec{r}_{bf}(p)=r_{1}\hat{i}+r_{2}(1-p)^{2}\hat{j}+r_{3}(1-p)^{2}\hat{k},
\end{equation}
\begin{equation}
\vec{r}_{bpf}(p)=r_{1}(1-p)^{2}\hat{i}+r_{2}\hat{j}+r_{3}(1-p)^{2}\hat{k},
\end{equation}
which are obtained interchanging $r_{3}$ with $r_{1}$ and  $r_{3}$ with $r_{2}$, respectively, in the Phase Damping channel.

In Fig. \ref{fig:graficas1}  we also show the curve of Discord as a function
of Entanglement, for the Trace distance with initial condition $|r_{3}|<|r_{2}|<|r_{1}|$, exhibiting two sudden changes.

\begin{figure*} 
\includegraphics[width=0.40\textwidth]{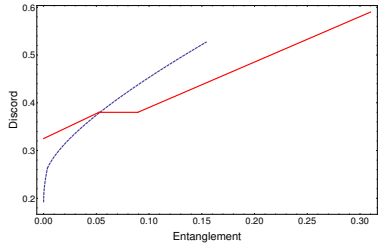}
\caption{(color online) Discord as a function of Entanglement for the initial state $r_{1}=0.65, r_{2}=0.59, r_{3}=-0.38$
evolved under local Phase Damping channel. The thick (red) line corresponds to the Trace distance and the
dashed (blue) line to the Hilber-Schmidt distance.}
\label{fig:graficas1}
\end{figure*}

\section{Entanglement-Discord relations for Bell-diagonal
states evolving under local Depolarizing channel}

\label{sec:EntDiscDEPO}
The Depolarizing channel describes the interaction of the system with
the environment that gets its state mixed up with the maximally entropic
state, with a probability $p$, i.e., $\rho\rightarrow(1-p)\rho+p(2^{-1}\mathbb{I}_{2}).$
This is equivalent to the possibility of occurrence of all the Pauli
errors with equal probability $p/4$. The Kraus operators for this
kind of noise channel are $K_{0}^{d}=\sqrt{1-3p/4}\mathbb{I}_{2}$,
$K_{1}^{d}=\sqrt{p/4}\sigma_{1}$, $K_{2}=\sqrt{p/4}\sigma_{2}$,
and $K_{3}=\sqrt{p/4}\sigma_{3}$. And the evolved correlation vector
is:
\begin{equation}
\vec{r}_{d}(p)=r_{1}(1-p)^{2}\hat{i}+r_{2}(1-p)^{2}\hat{j}+r_{3}(1-p)^{2}\hat{k}.
\end{equation}

\subsection{Entanglement - Discord relation using the Hilbert-Schmidt norm}
Applying the same reasoning we used above, we determine the Entanglement 
calculating the distance of the quantum state to the separable octahedron.
For this channel the Entanglement is expressed by
\begin{equation}
 E(p)=\frac{1}{3}\left((1-p)^2 \left(|r_{1}|+|r_{2}|+|r_{3}|\right)-1 \right)^2,
\end{equation}
the sudden death of Entanglement occurs for $p_{SD}= 1- \sqrt{\frac{1}{|r_{1}|+|r_{2}|+|r_{3}|}}.$\\
\label{sec:HSnormDepo}
For the Discord we have to determine the minimum
\begin{eqnarray}
 &&D_{i}(p)=(1-p)^4 \left(r_{j}^2 + r_{k}^2 \right ),   \,i,\,j,\,k =1,2,3, \nonumber \\ 
 &&i\neq j \neq k.
\end{eqnarray}
For this channel the Discord dynamics is very simple, never exhibiting
a sudden change which is obvious from its mathematical expression.\\
Once more, Discord and Entanglement can be analitically related:      
\begin{eqnarray}
 &&D_{i}(p)=\left(\frac{\sqrt{3E(p)}+1}{\left( |r_{1}|+|r_{2}|+|r_{3}| \right) } \right)^2\left( r_{j}^2 + r_{k}^2 \right), \nonumber \\
 &&i\neq j \neq k.
\end{eqnarray}
\subsection{Entanglement - Discord relation using the Trace norm}
\label{sec:TraceDistanceDepo}
As we discussed previously, the geometric Entanglement for the Bell-diagonal
states, measured by the Trace norm is the concurrence of the state. In this
case:
\begin{equation} 
C(p)= \frac{1}{2} \max \{0,|r_{1}\pm r_{2}|(1-p)^2 -\left(1 \pm r_{3} (1-p)^2 \right) \}
\label{eq:Entanglement0}
\end{equation}
The expression of the geometric Discord is
\begin{equation}
D(p)=|r_{int}|(1-p)^2
\end{equation}
and the relation between the geometric Discord and the geometric Entanglement is:
\begin{equation}
D(p)=|r_{int}|\frac{2C(p)+1}{|r_{1}\pm r_{2}| \mp r_{3}},
\end{equation}
where the $\pm$ depends on the maximum of Eq.(\ref{eq:Entanglement0}). 

In Fig. \ref{fig:graficasDepo} we show the curves of Discord as a function
of Entanglement obtained in this section for the Hilbert-Schmidt and the Trace distance.

\begin{figure*}
\includegraphics[width=0.40\textwidth]{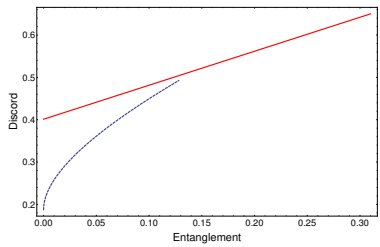}
\caption{(color online) Discord as a function of Entanglement for the initial state $r_{1}=0.65, r_{2}=0.59, r_{3}=-0.38$
evolved under local Depolarizing channel. The thick (red) line corresponds to the Trace distance and the
dashed (blue) line to the Hilber-Schmidt distance.}
\label{fig:graficasDepo}
\end{figure*}

\section{Concluding remarks}
\label{sec:conclusion}
 
 In this work, using a geometric approach, we obtained dynamical Entanglement-Discord relations for a family of initial Bell-diagonal states under the action of several local trace-preserving quantum-channels: Bit Flip, Bit Phase Flip, Phase Damping and Depolarizing.  It is worth pointing out that these channels keep the Bell-diagonal feature of the initial Bell-diagonal states, this made possible to work with closed analytic expressions.
A direct dynamical relation as the one presented here can be useful in several instances. As a matter of fact, whenever a result is derived for one of the quantities, we may apply it to simplify the description of the other one. An interesting example related to the obtention of an equation of motion for Discord in terms of the equation of motion for Entanglement can be considered. In the last few years, several works were dedicated to obtain factorization equations for Entanglement \cite{FactorLawEnt, Tiersch2009, ExpEntfactorlaw, OFarias1414, Yuan2013} , which can considerably simplify the experimental description of its dynamics. Via our relations, we can use those results in order to simplify the experimental description of the evolution of quantum Discord with time. \\
In addition, we obtained a very meaningful result working with the geometric measure of Entanglement in terms of the trace norm: we found that the trace distance between a general entangled X state and its closest separable state is the concurrence of the former. This result would allow to extend our study to more general evolutions that don't preserve the Bell-diagonal shape of the states, such as the Amplitude Damping channel. Some work in that direction is in progress.

\begin{acknowledgements}
JM aknowledges the financial support of the Brazilian funding agencies: Conselho Nacional de Desenvolvimento Cient\'ifico e Tecnol\'ogico (CNPq), process 441875/2014-9, Instituto Nacional de Ci\^encia e Tecnologia de Informa\c{c}\~ao Qu\^antica (INCT-IQ), process 2008/57856-6, and Coordena\c{c}\~ao de Desenvolvimento de Pessoal de N\'{i}vel Superior (CAPES), process 6531/2014-08. JM thanks the hospitality of the Physics Institute and Laser Spectroscopy Group at the Universidad de la Rep\'{u}blica, Uruguay. 
\end{acknowledgements}


%
%

\end{document}